\documentclass[journal,10pt,final]{IEEEtran}

\usepackage{amsmath,graphicx,amsthm}
\usepackage{comment,amsmath,amssymb,epsfig,amsfonts,verbatim, subfigure}
\usepackage{psfrag}
\usepackage{graphicx,times}
\usepackage{algorithmic}
\usepackage[boxed]{algorithm}
\usepackage{cite}
\usepackage{units}
\usepackage[usenames,dvipsnames,svgnames,table]{xcolor}
\usepackage{mathrsfs}
\usepackage{stfloats}
\usepackage{subfigure}
\usepackage{fixltx2e}

\graphicspath{{./Figures/}}

\newtheorem{theorem}{Theorem}

\newtheorem{lemma}{Lemma}

\theoremstyle{definition}
\newtheorem{definition}{Definition}

\usepackage{tikz}
\usetikzlibrary{arrows,shapes}
\usetikzlibrary{positioning}

\usepackage{pgfplots}
\usepgfplotslibrary{fillbetween}
\pgfplotsset{compat=newest}
\newlength\figureheight
\newlength\figurewidth
\definecolor{myGreen}{rgb}{0,0.5,0}

\title{Rate region boundary of the SISO Z-interference channel with improper signaling}
\author{\IEEEauthorblockN{Christian Lameiro\authorrefmark{1},~\IEEEmembership{Member,~IEEE,} Ignacio Santamar\'ia\authorrefmark{2},~\IEEEmembership{Senior~Member,~IEEE,} and Peter J. Schreier\authorrefmark{1},~\IEEEmembership{Senior~Member,~IEEE}\\}
\thanks{\authorrefmark{1}C. Lameiro  and P. J. Schreier are with the Signal \& System Theory Group, Universit\"{a}t Paderborn, Germany (email: \{christian.lameiro, peter.schreier\}@sst.upb.de).}
\thanks{\authorrefmark{2}I. Santamar\'ia is with the Department of Communications Engineering,
University of Cantabria, Spain (e-mail: i.santamaria@unican.es).}}

\begin{document}
\maketitle
\begin{abstract}
This paper provides a complete characterization of the boundary of an achievable rate region, called the Pareto boundary, of the single-antenna Z interference channel (Z-IC), when interference is treated as noise and users transmit complex Gaussian signals that are allowed to be improper. By considering the augmented complex formulation, we derive a necessary and sufficient condition for improper signaling to be optimal. This condition is stated as a threshold on the interference channel coefficient, which is a function of the interfered user rate and which allows insightful interpretations into the behavior of the achievable rates in terms of the circularity coefficient (i.e., degree of impropriety). Furthermore, the optimal circularity coefficient is provided in closed form. The simplicity of the obtained characterization permits interesting insights into when and how improper signaling outperforms proper signaling in the single-antenna Z-IC. We also provide an in-depth discussion on the optimal strategies and the properties of the Pareto boundary.
\end{abstract}
\section{Introduction}
\label{sec:intro}
It is widely known that proper Gaussian signals are capacity-achieving in different wireless communication networks, such as the point-to-point, broadcast and multiple-access channels. Because of that, the use of such a signaling scheme is generally assumed in the study of multiuser wireless networks. The capacity-achieving property of proper signaling stems from the maximum entropy theorem, which states that the entropy of a random variable under a power constraint is maximized for a proper Gaussian distribution \cite{Neeser1993}. However, in networks where interference presents the major limiting factor, proper Gaussian signaling has recently been shown to be suboptimal, and improper Gaussian signaling, also known as asymmetric complex signaling, has been proved to outperform proper signaling in different interference networks \cite{Cadambe2010,LameiroSantamaria:2013:Degrees-of-Freedom-for-the-4-User-SISO-Interference,Wang2011,Shin2012,Yang2014,Ho2012,Zeng2013,Zeng2013twc,Nguyen2015,LagenMorancho2015,Hellings2013,Kim2012,LameiroSantamariaSchreier:2015:Benefits-of-Improper-Signaling-for-Underlay,LameiroSantamariaSchreier:2015:Analysis-of-maximally-improper-signalling,Amin2015,Kurniawan2015,Lagen2014,Lagen2016}.

An improper complex random variable differs from its proper counterpart in that its real and imaginary parts are correlated or have unequal variance, or, in other words, the random variable is correlated with its complex conjugate \cite{Schreier2010}. Such signals arise naturally in communications, e.g., due to gain imbalance between the in-phase and in-quadrature branches, or due to the use of specific digital modulations, such as binary phase shift keying (BPSK) or Gaussian minimum shift keying (GMSK). Whenever the received signal is improper, linear operations must be replaced by widely linear operations, which are linear in both the random variable and its complex conjugate, in order to fully exploit the correlation between the signal and its complex conjugate \cite{Schreier2010,Adali2011,Taubock2012}. The design of such widely linear receivers has been extensively studied in the literature (see, e.g., \cite{Gerstacker2003,Buzzi2006,Chevalier2006,Schober2004} and references therein). However, the transmission of improper signals to handle interference more effectively is a rather new line of research. 

We would like to add that digital modulation schemes yield cyclostationary signals, i.e., the mean, autocovariance and complementary autocovariance funtions are periodic. It has been shown that exploiting this property along with the impropriety leads to an improved performance \cite{Han2012,Yeo2014,Yeo2015,Yeo2015_2,Kim2016}. For example, \cite{Han2012} shows that the optimal transmit signal must also be cyclostationary if the receiver is corrupted by a cyclostationary Gaussian noise. 

The first study on the benefits of improper signaling for interference management was carried out in the 3-user interference channel (IC) \cite{Cadambe2010}. That work showed an improvement in terms of degrees-of-freedom (DoF), which represent the maximum number of interference-free streams and characterize the asymptotic sum-capacity. Similar DoF results were derived for the 4-user IC\cite{LameiroSantamaria:2013:Degrees-of-Freedom-for-the-4-User-SISO-Interference}, the 3-user multiple-input multiple-output (MIMO) IC \cite{Wang2011}, the interfering broadcast channel \cite{Shin2012}, and the MIMO X-channel \cite{Yang2014}. However, improper signaling not only increases the achievable DoF, but also the achievable rates in interference-limited networks. The optimal rate region boundary for maximally improper transmissions (i.e., perfect correlation between real and imaginary parts or either zero real or imagniary part) was derived for the 2-user IC in \cite{Ho2012}, showing substantial improvements over proper signaling. Additionally, \cite{Zeng2013} proposed a suboptimal design of the improper transmit parameters, which outperforms the proper and the maximally improper scheme. A similar suboptimal design was also proposed in \cite{Zeng2013twc} for the $K$-user multiple-input single-output (MISO) IC. Improper signaling in the IC has also been applied to reduce the symbol error rate \cite{Nguyen2015} and as a mixed improper/proper approach in the MIMO-IC \cite{LagenMorancho2015}. In addition to the IC, the use of improper Gaussian signaling has also been shown beneficial for other multiuser scenarios, such as the broadcast channel with linear precoding \cite{Hellings2013}, relay-assisted communications \cite{Kim2012}, or underlay cognitive radio networks \cite{LameiroSantamariaSchreier:2015:Benefits-of-Improper-Signaling-for-Underlay,LameiroSantamariaSchreier:2015:Analysis-of-maximally-improper-signalling,Amin2015}.

A particular case of the 2-user IC is the Z-IC, also known as one-sided IC \cite{Costa1985}. The difference with respect to the 2-user IC is the fact that only one of the receivers is affected by interference.\footnote{Notice that the Z-IC is different from the Z-channel, where the cross-link also conveys a desired message (see, e.g., \cite{Prasad2015} and references therein)} The capacity region of the Z-IC is only known in the strong and very strong interference regimes \cite{Sato1981,Costa1985}, and it is achievable by a non-linear operation at the receiver. Less complex non-linear techniques have also been studied (see, e.g., the characterization of the rate region boundary of the 2-user IC with successive interference cancellation at the receivers \cite{SongChoi2014}). Even for such techniques, it is not known whether proper signaling is optimal. Nevertheless, it is more convenient to perform linear (or widely-linear) operations while treating the interference as noise, so that the complexity can be reduced. Restricted to linear operations, improper signaling presents a useful tool to improve the performance over proper signaling. 

Improper signaling for the Z-IC has recently been considered in \cite{Kurniawan2015}, where the sum-rate maximizing scheme was derived in closed form. To that end, \cite{Kurniawan2015} considered the so-called real-composite model, where complex signals are regarded as real signals of double dimension. However, despite some remarkable efforts \cite{Hellings2015}, the real-composite model is not as insightful as the augmented complex model, which works with the signal and its complex conjugate. For example, the circularity coefficient, which measures the degree of impropriety, is a quantity easily derived in the augmented complex formulation, but is much more difficult to express through its real-composite counterpart. The multi-antenna Z-IC with improper signaling has been addressed in \cite{Lagen2014} and its journal version \cite{Lagen2016}. The inclusion of the spatial dimension makes a complete analytical assessment intractable, which is why the authors proposed an heuristic scheme to optimize the widely linear operation at the transmitter, which permits a trade-off between the rates of both users. This way, \cite{Lagen2016} obtained an achievable rate region that is larger than that obtained by proper signaling. Although the authors of \cite{Lagen2014,Lagen2016} mainly focused on the real-composite representation, they also considered the augmented complex formulation. Thus, the optimal transmission scheme for the interfered user is obtained using the former, whereas that of the interfering user through the latter.

In our work, we adopt the augmented complex model to provide a complete and insightful characterization of the optimal rate region boundary, called the Pareto boundary, of the single-antenna Z-IC, when users may transmit improper Gaussian signals, assuming that interference is treated as noise. Our main contributions are summarized next.

\begin{itemize}
	\item We extend the results of \cite{Kurniawan2015}, where only one point of the rate region boundary is derived, to provide a complete characterization of the Pareto optimal boundary in closed form. We show that the rate region boundary can be described by a threshold on the interference channel coefficient, which determines when improper signaling is optimal.
	\item By adopting the augmented complex formulation, we provide, for each point of the boundary, closed-form expressions for the transmit powers and circularity coefficients, which are a direct measure of the degree of impropriety of the transmit signals. This permits insightful conclusions and a full assessment of the improvements of improper signaling over proper signaling in the single-antenna Z-IC. Thus, we analyze how the degree of impropriety affects the rate in the different boundary points, and we investigate the conditions that must be fulfilled for improper signaling to outperform proper signaling. The connection between the optimal circularity coefficients of both users and a further in-depth discussion of our characterization is also provided.
\end{itemize}

The rest of the paper is organized as follows. Section \ref{sec:model} provides some preliminaries of improper random variables and describes the system model. The characterization of the rate region boundary is derived in Section \ref{sec:main}, and a discussion on the results is presented in Section \ref{sec:dis} along with several numerical examples illustrating our findings. Finally, Section \ref{sec:conc} concludes the paper.

\section{System model}\label{sec:model}
\subsection{Preliminaries of improper complex random variables}\label{sec:modelPre}
We first provide some definitions and results for improper random variables that will be used throughout the paper. For a comprehensive treatment of the subject, we refer the reader to \cite{Schreier2010}.

The \emph{variance} of a zero-mean complex random variable $x$ is defined as $\sigma
^2=\operatorname{E}[|x|^2]$, where $|\cdot|$ is the absolute value and $\operatorname{E}[\cdot]$ is the expectation operator. The \emph{complementary variance} of a zero-mean complex random variable $x$ is defined as $\tilde{\sigma}^2=\operatorname{E}[x^2]$. If $\tilde{\sigma}^2=0$, then $x$ is called \emph{proper}, otherwise \emph{improper}. Furthermore, $\sigma^2$ and $\tilde{\sigma}^2$ are a valid pair of variance and complementary variance if and only if $\sigma^2\geq0$ and $|\tilde{\sigma}^2|\leq\sigma^2$.

The \emph{circularity coefficient} of a complex random variable $x$ is defined as the absolute value of the quotient of its complementary variance and its variance, i.e.,
\begin{equation}
	\kappa=\frac{\left|\tilde{\sigma}^2\right|}{\sigma^2} \; .
\end{equation}
The circularity coefficient satisfies $0\leq\kappa\leq1$ and thus measures the degree of impropriety of $x$. If $\kappa=0$, then $x$ is \emph{proper}, otherwise \emph{improper}. If $\kappa=1$ we call $x$ \emph{maximally improper}.

\subsection{System description}
\begin{figure}[t!]
\centering
\includegraphics[width=1\columnwidth]{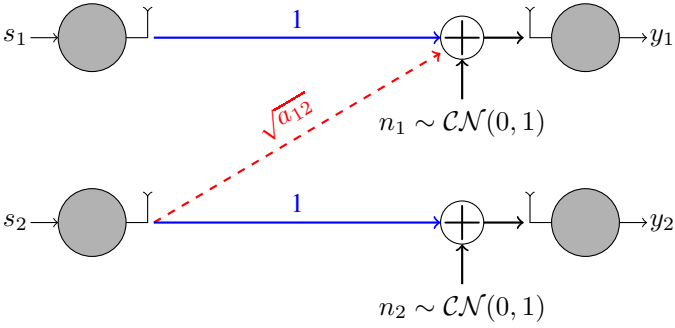}
\caption{SISO Z-IC in standard form. This model is described by three parameters: the power budgets $P_1$ and $P_2$, and the interference channel coefficient $a_{12}$.}
\label{fig:ZIC}
\end{figure}
We consider the single-input single-output (SISO) Z-IC with no symbols extensions. Without loss of generality and for the sake of exposition, we adhere to the standard form, as depicted in Fig. \ref{fig:ZIC}. Denoting by $\sqrt{a_{12}}$ the real channel coefficient between transmitter $2$ and receiver $1$, the signal at both receivers can be modeled by
\begin{align}
	y_1&= s_1+\sqrt{a_{12}} s_2+n_1 \; ,\\
	y_2&= s_2+n_2 \; ,
\end{align}
where $s_i$ and $n_i$ are the transmitted signal and noise of the $i$th user, respectively. The additive white Gaussian noise (AWGN) has variance 1 and is assumed to be proper, whereas the transmitted signals are complex Gaussian random variables with variance $\operatorname{E}[|s_i|^2]=p_i$  and complementary variance $\operatorname{E}[s_i^2]=\tilde{p}_i$. Thus, the rate achieved by each user, as a function of the design parameters $p_i$ and $\tilde{p}_i$, $i=1,2$, is given by \cite{Zeng2013}
\begin{align}
	R_1\left(p_1,\tilde{p}_1,p_2,\tilde{p}_2\right)=&\log_2\left(1+\frac{p_1}{1+p_2a_{12}}\right)\notag\\
	&+\frac{1}{2}\log_2\left(\frac{1-c_{y_1}^{-2}|\tilde{c}_{y_1}|^2}{1-c_{z_1}^{-2}|\tilde{c}_{z_1}|^2}\right) \; ,\label{eq:R1init}
\end{align}
\begin{equation}
	R_2\left(p_2,\tilde{p}_2\right)=\log_2\left(1+p_2\right)+\frac{1}{2}\log_2\left(\frac{1-c_{y_2}^{-2}|\tilde{c}_{y_2}|^2}{1-c_{z_2}^{-2}|\tilde{c}_{z_2}|^2}\right) \; , \label{eq:R2init}
\end{equation}
where
\begin{align}
	c_{y_1}&=p_1+p_2a_{12}+1 \; ,\\
	c_{y_2}&=p_2+1	 \; ,
\end{align}
are the variances of the received signals,
\begin{align}
	\tilde{c}_{y_1}&=\tilde{p}_1+\tilde{p}_2a_{12} \; ,\\
	\tilde{c}_{y_2}&=\tilde{p}_2 \; ,
\end{align}
the complementary variances of the received signals, 
\begin{align}
	c_{z_1}&=p_2a_{12}+1 \; ,\\
	c_{z_2}&=1 \; ,
\end{align}
are the variances of the interference-plus-noise signals, $z_1=\sqrt{a_{12}}s_2+n_1$, $z_2=n_2$, and
\begin{align}
	\tilde{c}_{z_1}&=\tilde{p}_2a_{12} \; ,\\
	\tilde{c}_{z_2}&=0 \; ,
\end{align}
the complementary variances of $z_1$ and $z_2$. Assuming that the power budget of the $i$th user is $P_i$, the achievable rate region with improper Gaussian signaling is then the union of all achievable rate tuples, i.e.,
\begin{equation}\label{eq:RateRegion}
	\mathcal{R}=\underset{\substack{0\leq p_i\leq P_i \\ |\tilde{p}_i|\leq p_i}\, , \, \forall i}{\bigcup}\left(R_1\left(p_1,\tilde{p}_1,p_2,\tilde{p}_2\right),R_2\left(p_2,\tilde{p}_2\right)\right) \; .
\end{equation} 
Notice that we have included the constraint $|\tilde{p}_i|\leq p_i$ in \eqref{eq:RateRegion}. As stated in Section \ref{sec:modelPre}, this condition must be fulfilled for $p_i$ and $\tilde{p}_i$ to be a valid pair of variance and complementary variance.

\section{Pareto boundary of the rate region}\label{sec:main}
The Pareto boundary of the rate region described by \eqref{eq:RateRegion} comprises all Pareto optimal points, which are defined as follows \cite{Jorswieck2008}.
\begin{definition}
	We call the rate pair $(R_1,R_2)$ Pareto-optimal if $(R_1',R_2)$ and $(R_1,R_2')$, with $R_1'>R_1$ and $R_2'>R_2$, are not achievable.\footnote{For the sake of brevity, we omit the dependence of $R_1$ and $R_2$ on the design parameters when it is self-evident or not relevant.} 
\end{definition}
In this section we characterize this boundary by deriving the optimal transmission parameters, $p_i$ and $\tilde{p}_i$, $i=1,2$, that achieve each point of the boundary.

First, we notice that, since user 1 does not interfere with user 2, its optimal transmit strategy maximizes its own achievable rate. Consequently, its transmit power must be maximized, which implies $p_1=P_1$. Second, $p_i$ and $\tilde{p}_i$ are a valid pair of variance and complementary variance if and only if $p_i\geq0$ and $|\tilde{p}_i|\leq p_i$. Consequently, the complementary variance can be expressed as $\tilde{p}_i=p_i\kappa_ie^{\jmath\phi_i}$, where $\kappa_i$ is the circularity coefficient, which measures the degree of impropriety. Hence, $|\tilde{p}_i|\leq p_i$ is equivalent to $0\leq\kappa_i\leq1$. With these considerations, $R_1$ can then be expressed as
\begin{align}
	 R_1=\frac{1}{2}\log_2&\left[\frac{\left(p_2a_{12}+P_1+1\right)^2}{\left(p_2a_{12}+1\right)^2-\left|p_2e^{\jmath\phi_2}\kappa_2a_{12}\right|^2}\right.\notag\\
	 &\left.-\frac{\left|p_2e^{\jmath\phi_2}\kappa_2a_{12}+P_1e^{\jmath\phi_1}\kappa_1\right|^2}{\left(p_2a_{12}+1\right)^2-\left|p_2e^{\jmath\phi_2}\kappa_2a_{12}\right|^2}\right] \; .\label{eq:R1}
\end{align}
Through \eqref{eq:R1} it is clear that $R_1$ is maximized when $\left|p_2e^{\jmath\phi_2}\kappa_2a_{12}+P_1e^{\jmath\phi_1}\kappa_1\right|^2$ is minimized, which yields
\begin{align}
	\kappa_1&=\min\left(\frac{p_2\kappa_2a_{12}}{P_1},1\right) \; ,\label{eq:kappa1}\\
	\phi_1&=\phi_2+\pi \; .\label{eq:phi1}
\end{align}
From \eqref{eq:kappa1} we observe that, if $\kappa_2=0$, i.e., user 2 transmits a proper signal, then user 1 must also transmit a proper signal by setting $\kappa_1=0$. Similarly, if user 2 transmits an improper signal ($\kappa_2>0$), then the signal transmitted by user 1 must also be improper. According to \eqref{eq:phi1}, the difference between the phases of the complementary variances of the desired and interference signals at receiver 1 is $\pi$. Such a phase difference can be interpreted by looking at the joint distribution of the real and imaginary parts of the desired signal and interference at receiver 1. The level contours of their distributions are ellipses whose major axes are rotated by $\pi/2$ with respect to each other \cite{Schreier2010}, so that the signal and interference power are concentrated along orthogonal dimensions.

Now we observe the following. With the optimal choice of $\phi_1$, given by \eqref{eq:phi1}, the effect of $\phi_2$ is compensated at receiver 1. Thus the achievable rate of user 1 is independent of the specific value of $\phi_2$. Furthermore, since user 2 is not affected by interference, $\phi_2$ also has no impact on its achievable rate. Hence, without loss of generality, we can take $\phi_2=0$. With these considerations, the design parameters are reduced to the transmit power $p_2$ and circularity coefficient $\kappa_2$ of user 2. After some manipulations of \eqref{eq:R2init} and \eqref{eq:R1}, the achievable rates of user 1 and user 2, as a function of the design parameters, are given by  
\begin{align}
	 R_1\left(p_2,\kappa_2\right)\hspace{-0.1cm}&=\hspace{-0.1cm}\left\{\begin{matrix}\frac{1}{2}\log_2\left[\frac{\left(p_2a_{12}+P_1+1\right)^2}{1+p_2a_{12}\left(p_2a_{12}(1-\kappa_2^2)+2\right)}\right]  &\hspace{-0.4cm} \text{if} \; \kappa_1<1\\
	 \frac{1}{2}\log_2\left[1+\frac{2P_1\left(p_2a_{12}(1+\kappa_2)+1\right)}{1+p_2a_{12}\left(p_2a_{12}(1-\kappa_2^2)+2\right)}\right]  & \hspace{-0.4cm} \text{if} \; \kappa_1=1\end{matrix}\right. \hspace{-0.1cm}, \label{eq:R1pk}\\
	R_2\left(p_2,\kappa_2\right)&=\frac{1}{2}\log_2\left[1+p_2\left(p_2(1-\kappa_2^2)+2\right)\right] \; ,\label{eq:R2pk}
\end{align}
and the achievable rate region defined in \eqref{eq:RateRegion} can then be expressed as
\begin{equation}\label{eq:RateRegionpk}
	\mathcal{R}=\underset{\substack{0\leq p_2\leq P_2 \\ 0\leq\kappa_2\leq1}}{\bigcup}\left(R_1\left(p_2,\kappa_2\right),R_2\left(p_2,\kappa_2\right)\right) \; .
\end{equation}
In order to characterize the boundary of the region defined in \eqref{eq:RateRegionpk}, we notice that the achievable rate of user 1 is bounded as
\begin{equation}\label{eq:R1lbub}
	0\leq R_1\left(p_2,\kappa_2\right)\leq\log_2\left(1+P_1\right) \; .
\end{equation}
For each achievable rate of user 1, the corresponding Pareto optimal point is given by the one maximizing the rate of user 2, $R_2(p_2,\kappa_2)$, which can be cast as the following optimization problem
\begin{align}
\mathcal{P}:\hspace{0.5cm} & \underset{p_2,\kappa_2}{\text{maximize}}
& & R_2\left(p_2,\kappa_2\right) \; ,\notag\\
& \text{subject to}
& & 0\leq p_2 \leq P_2 \; , \notag\\
& & & 0\leq\kappa_2\leq1 \; , \notag\\
& & & R_1\left(p_2,\kappa_2\right)\geq\alpha\log_2\left(1+P_1\right) \; ,\label{eq:RateConst}
\end{align}
for a given $\alpha\in[0,1]$. Thus, we can compute every point of the rate region boundary by varying $\alpha$ between 0 and 1 and solving problem $\mathcal{P}$.

The set of constraints of problem $\mathcal{P}$, which defines the feasibility set of our design parameters, consists of two constraints affecting the design parameters independently, namely, the power budget constraint and the bounds on the circularity coefficient, and an additional one that jointly constrains $p_2$ and $\kappa_2$. The latter expresses a rate constraint on user 1, so that a specific point of the region boundary, determined by $\alpha$, is computed. For a given $\kappa_2$, this constraint essentially limits the transmit power of user 2, $p_2$. Consequently, we can rewrite it in a more convenient form, as we express in the following lemma.
\begin{lemma}\label{th:lemma0}
	Let $\gamma_x=2^x-1$ and $\bar{R}=\alpha\log_2(1+P_1)$. The rate constraint $R_1\left(p_2,\kappa_2\right)\geq\bar{R}$ is then equivalent to the power constraint $p_2\leq q\left(\kappa_2\right)$, where $q(\kappa_2)$ is given by
\begin{equation}\label{eq:pk}
	 q(\kappa_2)\hspace{-0.1cm}=\hspace{-0.1cm}\left\{\begin{matrix}\frac{\left(P_1-\gamma_{2\bar{R}}\right)+\sqrt{\left(\gamma_{2\bar{R}}+1\right)\left[P_1^2\left(1-\kappa_2^2\right)+\left(\gamma_{2\bar{R}}-2P_1\right)\kappa_2^2\right]}}{a_{12}\left[\left(\gamma_{2\bar{R}}+1\right)\left(1-\kappa_2^2\right)-1\right]}  & \hspace{-0.3cm}\text{if} \; \kappa_1<1 \\
	 \frac{1}{a_{12}\left(1-\kappa_2\right)}\left(\frac{2P_1}{\gamma_{2\bar{R}}}-1\right) & \hspace{-0.3cm} \text{if} \; \kappa_1=1\end{matrix}\right. \hspace{-0.1cm}.
\end{equation}
\end{lemma}
\begin{proof}
	Please refer to Appendix \ref{app:lemma0}.
\end{proof}
As a result of Lemma \ref{th:lemma0}, we can equivalently state problem $\mathcal{P}$ as
\begin{align*}
\mathcal{P}:\hspace{0.5cm} & \underset{p_2,\kappa_2}{\text{maximize}}
& & R_2\left(p_2,\kappa_2\right) \; ,\\
& \text{subject to}
& & 0\leq p_2 \leq\min\left[q(\kappa_2),P_2\right] \; , \\
& & & 0\leq\kappa_2\leq1 \; .
\end{align*}

\begin{figure}[t!]
\centering
\includegraphics[width=1\columnwidth]{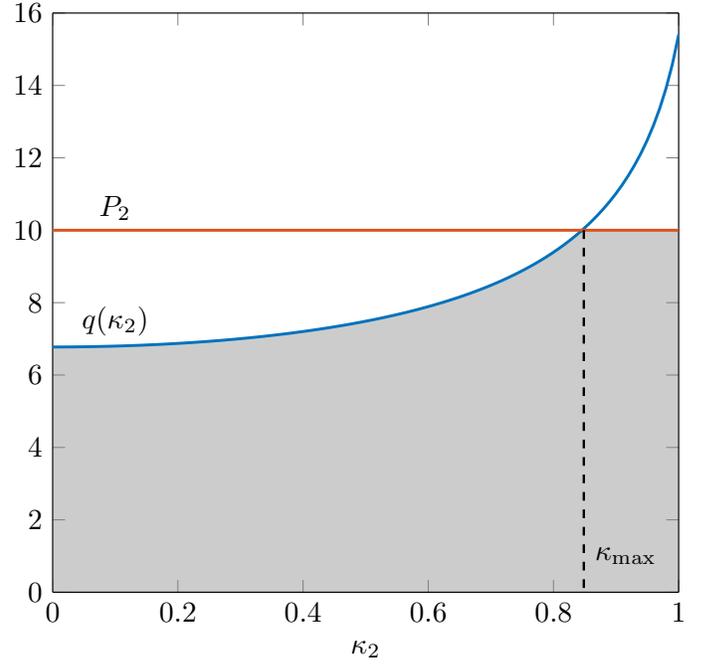}
\caption{Example of the transmit power and power constraints of user 2, for $P_1=20$, $P_2=10$, $a_{12}=0.5$ and $\alpha=0.7$. The shaded area is the set of possible transmit powers, and $\kappa_{\max}$ is the value of $\kappa_2$ such that $q(\kappa_{\max})=P_2$.}
\label{fig:pk}
\end{figure}
For the sake of illustration, we plot in Fig. \ref{fig:pk} an example of the two constraints affecting the transmit power of user 2, namely, $q(\kappa_2)$ and $P_2$. Obviously, $q(\kappa_2)$ is increasing in $\kappa_2$, since an interference with a higher degree of impropriety is less harmful, which is why user 1 tolerates a higher amount of interference power without reducing its achievable rate. To achieve the global maximum of problem $\mathcal{P}$, it is clear that we must set $p_2(\kappa_2)=\min[q(\kappa_2),P_2]$, where we have explicitly expressed its dependence on $\kappa_2$. Consequently, the number of design parameters is reduced to one, $\kappa_2$, and $\mathcal{P}$ is further simplified to
\begin{align*}
\mathcal{P}:\hspace{0.5cm} & \underset{\kappa_2}{\text{maximize}}
& & R_2\left(\kappa_2\right) \; ,\\
& \text{subject to}
& &  0\leq\kappa_2\leq1 \; ,
\end{align*}
where $R_2(\kappa_2)$ can now be expressed as
\begin{equation}\label{eq:R2}
	R_2\left(\kappa_2\right)=\frac{1}{2}\log_2\left\{1+p_2\left(\kappa_2\right)\left[p_2\left(\kappa_2\right)\left(1-\kappa_2^2\right)+2\right]\right\} \; .
\end{equation}
Notice that, by expressing $p_2$ as a function of $\kappa_2$, $R_2(\kappa_2)$ is now also a function of $\kappa_2$ only. That is, the key task now is to determine the optimal circularity coefficient of the second user, or, in other words, the degree of impropriety of its transmit signal such that its achievable rate, given by \eqref{eq:R2}, is maximized.

In the forthcoming lines we will analyze when $R_2(\kappa_2)$ is maximized by an improper signal, i.e., $\kappa_2>0$, and then we will derive the optimal value of $\kappa_2$ in those cases. That is, we want to determine the conditions that must be fulfilled for improper signaling to outperform conventional proper signaling. Since there are two different power constraints affecting $p_2$, namely, the power budget of user 2 and the interference power created at user 1, we will start by dropping the power budget constraint to analyze how the interference constraint, $q(\kappa_2)$, affects the rate of user 2 as a function of $\kappa_2$. We first present the following lemma.
\begin{lemma}\label{th:lemma1}
	Let $p_2(\kappa_2)=q(\kappa_2)$ and assume that there exists $\hat{\kappa}_2$ such that $\left.\frac{\partial R_2(\kappa_2)}{\partial\kappa_2^2}\right|_{\kappa_2=\hat{\kappa}_2}\geq0$. Then $\frac{\partial R_2(\kappa_2)}{\partial\kappa_2^2}>0$ for all $\kappa_2>\hat{\kappa}_2$.
\end{lemma}
\begin{proof}
	Please refer to Appendix \ref{app:lemma1}.	
\end{proof}
Lemma \ref{th:lemma1} leads to the following key result.
\begin{lemma}\label{th:lemma2}
	Let $p_2(\kappa_2)=q(\kappa_2)$ and
	\begin{align}
	\mu(\alpha)&=1-\frac{P_1}{\gamma_{2\bar{R}}-\gamma_{\bar{R}}} \; ,\label{eq:mu1}\\
	\iota(\alpha)&=\left\{\begin{matrix}\frac{\left(P_1-\gamma_{\bar{R}}\right)^2}{\left(\gamma_{2\bar{R}}-\gamma_{\bar{R}}\right)\left(\sqrt{\gamma_{2\bar{R}}-2P_1}-\gamma_{\bar{R}}\right)^2} & \text{if} \; 2P_1<\gamma_{2\bar{R}} \\
0 & \text{otherwise}\end{matrix}\right. \; .\label{eq:mu2}
\end{align}
Then, the dependency of $R_2(\kappa_2)$ on $\kappa_2$ can be described as follows.
\begin{itemize}
	\item If $a_{12}\geq\mu(\alpha)$, then $R_2(\kappa_2)$  increases monotonically in $\kappa_2$.
	\item If $\mu(\alpha)>a_{12}>\iota(\alpha)$, then there exists $\tilde{\kappa}_2>0$ such that $R_2(\kappa_2)\leq R_2(0)$ for $\kappa_2\leq\tilde{\kappa}_2$ and $R_2(\kappa_2)>R_2(0)$ otherwise. Furthermore, $R_2(\kappa_2)$  increases monotonically for $\kappa_2\geq\tilde{\kappa}_2$.
	\item If $\iota(\alpha)\geq a_{12}$, then $R_2(\kappa_2)\leq R_2(0)$ for all values of $\kappa_2$.
\end{itemize}
\end{lemma}
\begin{proof}
	Please refer to Appendix \ref{app:lemma2}.
\end{proof}

\begin{figure}[t!]
\centering
\includegraphics[width=1\columnwidth]{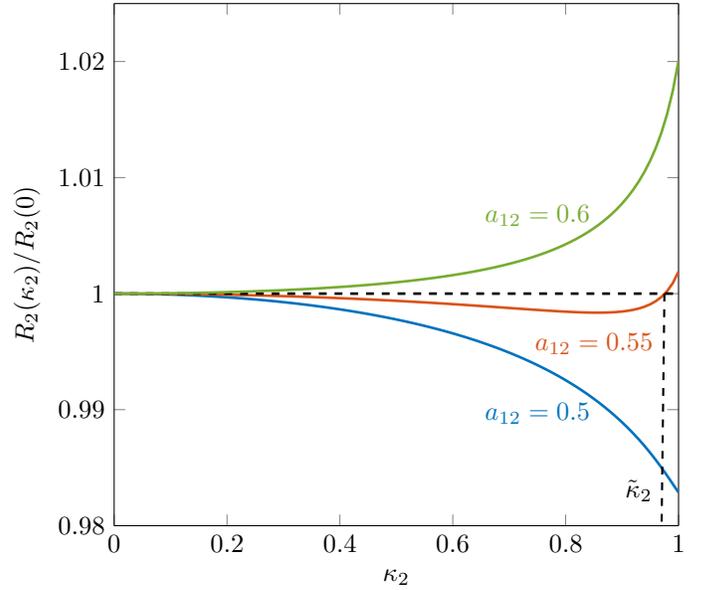}
\caption{Illustration of Lemma \ref{th:lemma2}. We set $\alpha=0.7$ and $P_1=10$, which yields $\mu(0.7)=0.571$ and $\iota(0.7)=0.545$.}
\label{fig:R2vsKappa}
\end{figure}
To illustrate the characterization provided in Lemma \ref{th:lemma2}, we plot in Fig. \ref{fig:R2vsKappa} the rate of user 2 normalized by the proper signaling rate with examples for each of the three cases described in Lemma \ref{th:lemma2}. We use $\alpha=0.7$ and $P_1=10$, which gives $\mu(0.7)=0.571$ and $\iota(0.7)=0.545$. Notice that the three curves vary only slightly with $\kappa_2$ because the three chosen values of $a_{12}$ 
(0.5, 0.55 and 0.6) are all close to the thresholds. Lemma \ref{th:lemma2} describes how the interference constraint $q(\kappa_2)$ shapes the dependency of the rate of user 2 on the degree of impropriety of its transmitted signal. Thus, if we drop the power budget constraint or, alternatively, if the power budget is sufficiently high, improper signaling outperforms proper signaling for all values of $\kappa_2$ if $a_{12}\geq\mu(\alpha)$, but only for some values of $\kappa_2$ if $\mu(\alpha)>a_{12}>\iota(\alpha)$, and it is always suboptimal if $a_{12}\leq\iota(\alpha)$. It can be clearly seen now what the impact of the power budget constraint is. For the cases $a_{12}\geq\mu(\alpha)$ and $\mu(\alpha)>a_{12}>\iota(\alpha)$, $\kappa_2=1$ maximizes $R_2(\kappa_2)$. However, increasing $\kappa_2$ is only meaningful if it also permits increasing the transmit power. When the power budget constraint is considered, there may exist $\kappa_{\max}$ such that $q(\kappa_{\max})=P_2$, hence increasing $\kappa_2$ beyond that point does not permit a further increase in the transmit power (see Fig. \ref{fig:pk}). Therefore, when $\mu(\alpha)>a_{12}>\iota(\alpha)$ improper signaling will be optimal if $\kappa_{\max}>\tilde{\kappa}_2$, where $\tilde{\kappa}_2>0$ is such that $R_2(\tilde{\kappa}_2)=R_2(0)$ (see Fig. \ref{fig:R2vsKappa}). On the other hand, improper signaling will be optimal for $a_{12}\geq\mu(\alpha)$ only if $q(0)<P_2$. This condition means that proper signaling does not permit transmitting at maximum power, thus there is still power left over, which can be exploited by improper signaling to improve the achievable rate.

With all these ingredients, we can derive a complete characterization of the optimality of improper signaling for this scenario and, consequently, of the Pareto-optimal region. Our main result is presented in the following theorem.
\begin{theorem}\label{th:theorem1}
	Let us define $\iota(\alpha)$ as in \eqref{eq:mu2} and
	\begin{equation}\label{eq:rho}
		\rho(\alpha)=\max\left[\frac{1}{P_2}\left(\frac{P_1}{\gamma_{\bar{R}}}-1\right),\nu(\alpha)\right] \; ,
	\end{equation}
	where
	\begin{equation}
		\nu(\alpha)=\left\{\begin{matrix}\eta(\alpha) & \text{if} \; \eta(\alpha)\geq\frac{P_1}{P_2}+a_I \\
		\frac{\gamma_{\bar{R}}\left(2\gamma_{2\bar{R}}+1\right)-P_1\left(2\gamma_{\bar{R}}+1\right)}{\gamma_{\bar{R}}\left[P_2+2\left(\gamma_{2\bar{R}}+1\right)\right]} & \text{otherwise}\; ,\end{matrix}\right.
	\end{equation}
	and
	\begin{align}
		\eta(\alpha)&=\frac{a_P-P_2a_I+\sqrt{(a_P-P_2a_I)^2+2P_2(a_I^2+a_P^2)}}{2} \; ,\\
		a_P&=\frac{1}{P_2}\left(\frac{P_1}{\gamma_{\bar{R}}}-1\right) \; ,\label{eq:aP}\\
		a_I&=\frac{1}{P_2}\left(\frac{2P_1}{\gamma_{2\bar{R}}}-1\right) \; .\label{eq:aI}
	\end{align}
	Then, improper signaling is required to maximize $R_2(\kappa_2)$ if and only if
	\begin{equation}
		a_{12}>\max\left[\iota(\alpha),\rho(\alpha)\right] \; .
	\end{equation}
	Furthermore, if this expression holds, the optimal circularity coefficient is
	\begin{equation}\label{eq:kOpt}
			\kappa_2=\left\{\begin{matrix} 1 & \text{if} \; q(1)\leq P_2 \\
		 \kappa_{\max} & \text{otherwise} \; ,\end{matrix} \right.
	\end{equation}
	where $\kappa_{\max}$ is the minimum value of $\kappa_2$ such that $P_2\leq q(\kappa_{\max})$.
\end{theorem}
\begin{proof}
	Please refer to Appendix \ref{app:theorem1}.
\end{proof}
Before concluding this section, we introduce the following definition, which will be useful to describe the properties of the Pareto boundary.
\begin{definition}
	We call \emph{power-limited region} all the points of the rate region boundary for which $\iota(\alpha)=0$ or, alternatively, $2P_1\geq\gamma_{2\bar{R}}$. We call the remaining points of the rate region the \emph{interference-limited region}, i.e., those for which $\iota(\alpha)>0$ or, alternatively, $2P_1<\gamma_{2\bar{R}}$.
\end{definition}

\section{Discussion and numerical examples}\label{sec:dis}
This section provides a discussion on the derived characterization along with some numerical examples illustrating the most remarkable features of improper signaling in the Z-IC. Afterwards, the connection to related works in the literature is presented.

\subsection{Optimal strategies}

\subsubsection{Optimality of proper signaling}\label{sec:optProp}
As pointed out at the beginning of Section \ref{sec:main}, if proper signaling is the optimal strategy for one of the users, then it is also optimal for the other one. This means that any point of the region boundary is achieved by both users employing the same signaling scheme, i.e., either proper or improper, but not a combination of both.

\subsubsection{Maximally improper signaling for both users is optimal at most at one boundary point}
It can be noticed that there is at most one boundary point where both users simultaneously transmit a maximally improper signal, i.e., $\kappa_1=\kappa_2=1$ happens for no more than one Pareto-optimal point. This is due to the fact that, if the rate constraint \eqref{eq:RateConst} can be fulfilled for $\kappa_1=1$ (which corresponds to the case $2P_1\geq\gamma_{2\bar{R}}$), then user 1 tolerates an infinite amount of maximally improper interference along the orthogonal direction (see second equation in \eqref{eq:pk}). However, user 2 may only increase its rate by increasing $\kappa_2$ if it is operating below its power budget, since the only purpose of increasing $\kappa_2$ is to increase $p_2$ as well. Because of this and because $q(\kappa_2)$ is a continuous function for $2P_1>\gamma_{2\bar{R}}$, setting $\kappa_1=\kappa_2=1$ is always suboptimal when $2P_1>\gamma_{2\bar{R}}$. However, this reasoning is not applicable when $2P_1=\gamma_{2\bar{R}}$, since $q(\kappa_2)$ turns into a discontinuous function for $\kappa_1=1$. This is because, in that case, $q(\kappa_2<1)=0$ and $q(\kappa_2=1)=\infty$. The intuition behind this behavior is that, when $2P_1=\gamma_{2\bar{R}}$, the first user achieves its corresponding rate with a maximally improper signal, i.e., $\kappa_1=1$, only if the interference is orthogonal to the signal subspace, i.e., only when the second user also transmits a maximally improper signal. Consequently, $\kappa_2=1$ must hold if $2P_1=\gamma_{2\bar{R}}$ and $\kappa_1=1$. Furthermore, since $p_2=P_2$ when $\kappa_2=1$, the condition $\kappa_1=1$ is equivalent, by \eqref{eq:kappa1}, to $P_2a_{12}\geq P_1$. That is, both users transmit maximally improper signals at the boundary point for which $R_1=\frac{1}{2}\log_2(1+2P_1)$ only if $P_2a_{12}\geq P_1$. 

To illustrate this property, we provide two simulation examples. As a first example we consider the channel coefficient $a_{12}=2$ and the power budgets $P_1=P_2=10$. Figure \ref{fig:CircCoeff} shows the optimal circularity coefficients and the transmit power of user 2. In this example $P_2a_{12}\geq P_1$ holds and because of that both users transmit maximally improper signals for $R_1=2.2\;\text{b/s/Hz}$. We can observe a discontinuity in the maximum transmit power of user 2, which will be explained later. Notice that, as $R_1$ approaches its maximum value of $3.46$ b/s/Hz, $\kappa_1$ goes towards 0 while $\kappa_2$ remains static at 1. Although this might seem to violate our statement in Section \ref{sec:optProp}, $\kappa_1$ is only equal to zero when $p_2$ is, which only happens at $R_1=3.46$ b/s/Hz. In this case, the circularity coefficient is no longer meaningful since the transmit power is zero. Therefore, all the boundary points satisfy our claim, which is also clear according to the optimal value of $\kappa_1$ in \eqref{eq:kappa1}. 
\begin{figure}[t!]
\centering
\includegraphics[width=1\columnwidth]{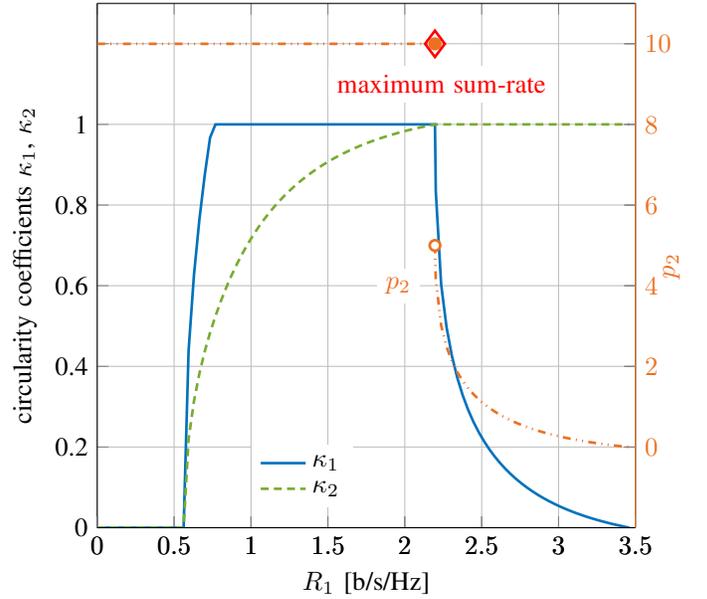}
\caption{Dependency of the optimal circularity coefficients and the transmit power $p_2$ on $R_1$ for $a_{12}=2$.}
\label{fig:CircCoeff}
\end{figure}

As a second example we consider $a_{12}=0.8$. In this case $P_2a_{12}<P_1$, hence it is expected that there are no boundary points that are achieved by both users transmitting maximally improper signals. This can be observed in Fig. \ref{fig:CircCoeff2}, which depicts the dependency of the circularity coefficients and transmit power on the rate, $R_1$. At the boundary point for $R_1=2.2\;\text{b/s/Hz}$, user 2 chooses the transmit signal as maximally improper, but the optimal circularity coefficient of user 1 equals 0.8. In Fig. \ref{fig:CircCoeff2} we can also observe a discontinuity in the maximum transmit power and in the circularity coefficients at approximately $R_1=2.9$ b/s/Hz. This discontinuity is different from the one observed in Fig. \ref{fig:CircCoeff} and can be explained as follows. At the point of the discontinuity we have $a_{12}=\max[\iota(\alpha),\rho(\alpha)]$. Furthermore, at this point $q(1)<P_2$. This means that maximally improper signaling achieves the same rate as proper signaling. Proper signaling starts outperforming improper signaling once $R_1$ increases beyond that point, in which case the optimal circularity coefficient jumps from $\kappa_2=1$ to $\kappa_2=0$, and the maximum transmit power jumps then from $q(1)$ to $q(0)$. 
\begin{figure}[t!]
\centering
\includegraphics[width=1\columnwidth]{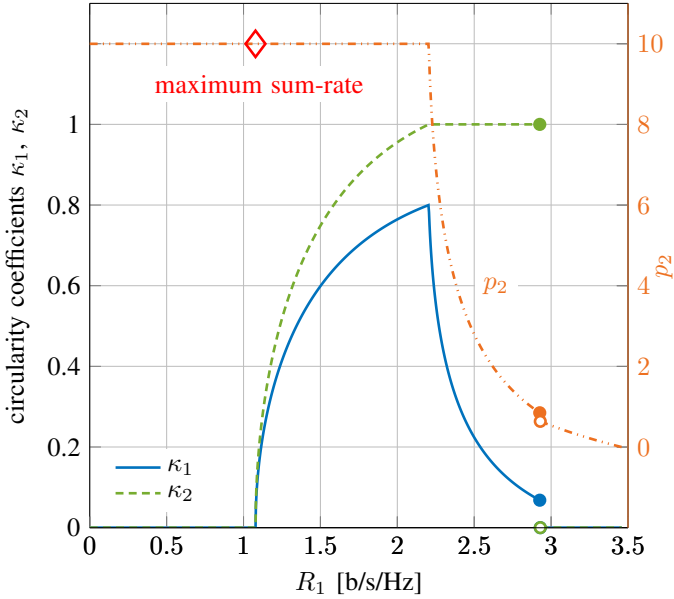}
\caption{Dependency of the optimal circularity coefficients and the transmit power $p_2$ on $R_1$ for $a_{12}=0.8$.}
\label{fig:CircCoeff2}
\end{figure}

\subsubsection{Relationship between $\kappa_1$ and $\kappa_2$}
Expression \eqref{eq:kappa1} also permits drawing insightful conclusions about the relationship between the circularity coefficients of both users. According to Theorem \ref{th:theorem1}, $0<\kappa_2<1$ implies $p_2=P_2$. Hence, $\kappa_1<\kappa_2$ holds if $\kappa_2>0$ and $P_2a_{12}<P_1$, i.e., when the signal-to-interference ratio (SIR) is greater than one. In such a case, it can be noticed that the signal transmitted by the first user is never chosen as maximally improper at any point of the Pareto boundary. This behavior can be clearly observed in Fig. \ref{fig:CircCoeff2}. On the other hand, when $P_2a_{12}\geq P_1$, or, alternatively, when the SIR is equal to or lower than one, $\kappa_1\geq\kappa_2$ holds whenever $0<\kappa_2<1$, as can be seen in Fig. \ref{fig:CircCoeff}. 

If the signal transmitted by user 1 is chosen as maximally improper for some points of the boundary, it will remain so as $\kappa_2\rightarrow1$, or, in other words, as $\gamma_{2\bar{R}}\rightarrow2P_1$. This is because $\kappa_2<1$ implies $p_2=P_2$, so that, by \eqref{eq:kappa1}, the first user will not decrease its circularity coefficient. However, once $\kappa_2$ equals 1, which corresponds to the point $2P_1=\gamma_{2\bar{R}}$, the degree of impropriety of the first user will then start decreasing with $R_1$, since these rates are not achievable for $\kappa_1=1$.

\subsection{Properties of the Pareto boundary}

\subsubsection{Behavior in the power-limited and interference-limited regions} 
The power-limited region, which refers to the power limitation of the second user, comprises all boundary points satisfying $2P_1\geq\gamma_{2\bar{R}}$, or, alternatively, $R_1\leq\frac{1}{2}\log_2(1+2P_1)$. In this region improper signaling is always optimal as long as the power budget $P_2$ is sufficiently high. In other words, the optimality of improper signaling is only determined by $\rho(\alpha)$, which depends on $P_2$ and goes towards 0 as $P_2$ increases. Furthermore, user 1 can achieve its required rate with a maximally improper signal, i.e., with $\kappa_1=1$. In this case, the interference along the unused dimension does not have any impact on its achievable rate, therefore $R_2\rightarrow\infty$ as $P_2\rightarrow\infty$. In the interference-limited region, user 1 must choose a circularity coefficient smaller than one to achieve the desired rate, and hence the tolerated interference is finite for all values of $\kappa_2$. As a result, the transmit power of the second user is eventually limited by $q(\kappa_2)$ as $P_2$ grows, which bounds the achievable rate. In other words, the optimality of improper signaling is eventually determined by $\iota(\alpha)$, which is independent of the power budget $P_2$.

\subsubsection{Transition between power-limited and interference-limited regions}
An interesting feature  is that there may be abrupt changes in the achievable rate of user 2 when we shift from one region to the other, which are due to a jump in the tolerable interference power (as observed in Fig. \ref{fig:CircCoeff}). This can be explained as follows. In the power-limited region the transmit power always equals the power budget when improper signaling is optimal, i.e., $p_2(\kappa_2)=P_2$ if $\kappa_2>0$. However, in the interference-limited region the transmit power is dominated by the function $q(\kappa_2)$ when the power budget exceeds the value of that function, i.e., when $P_2>q(1)$. Hence, by using \eqref{eq:kappa1}, it can be easily seen that $\lim_{\gamma_{2\bar{R}}\rightarrow2P_1^+}q_2(1)=\frac{P_1}{a_{12}}$, so that there may be a jump in the maximum transmit power, i.e., a discontinuity in the maximum transmit power of user 2 as a function of $R_1$, from $p_2=P_2$ to $p_2=\frac{P_1}{a_{12}}$ whenever $P_2>\frac{P_1}{a_{12}}$. That is, the lower the SIR, the more prominent the power jump is, whereas no jump will be observed when the SIR is equal to or greater than 1. This discontinuity in the maximum transmit power implies a similar jump in the maximum achievable rate of user 2, making $\lim_{\gamma_{2\bar{R}}\rightarrow2P_1^+}R_2(\kappa_2)\neq\lim_{\gamma_{2\bar{R}}\rightarrow2P_1^-}R_2(\kappa_2)$. This is illustrated in Figs. \ref{fig:RvsAlpha} and \ref{fig:RvsAlpha2}, where the Pareto boundary of the rate region is depicted for the previously considered examples, i.e., $a_{12}=2$ and $a_{12}=0.8$, respectively. We also depict in the figures the rate region boundary for proper signaling, and the enlargement of the rate region due to time sharing.
\begin{figure}[t!]
\centering
\includegraphics[width=1\columnwidth]{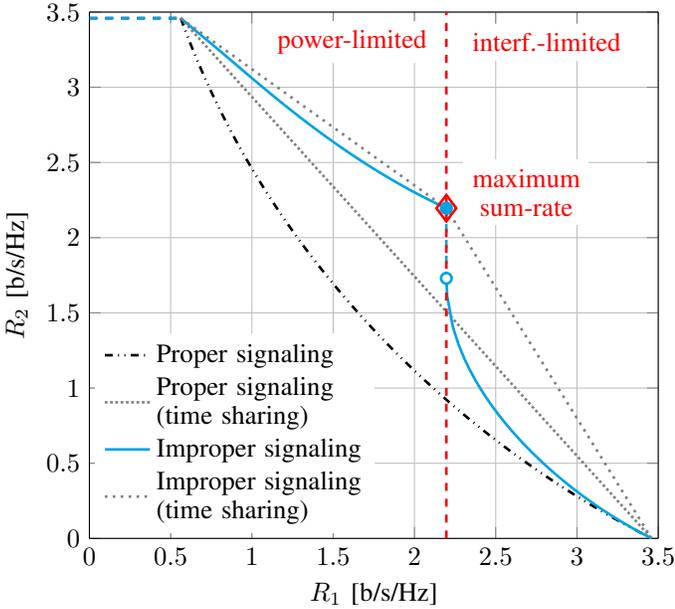}
\caption{Pareto boundary for proper and improper transmissions, with $P_1=P_2=10$ and $a_{12}=2$.}
\label{fig:RvsAlpha}
\end{figure}

Figure \ref{fig:RvsAlpha} corresponds to a scenario where the SIR is below one, therefore we can observe that the Pareto boundary is discontinuous. As explained earlier, this is the result of a jump in the maximum transmit power of user 2, $p_2$, which can also be observed in Fig. \ref{fig:CircCoeff}. Specifically, the maximum transmit power in the interference limited region equals $p_2=\frac{P_1}{a_{12}}=5$, so there is a power jump from $p_2=P_2=10$ to $p_2=5$. Notice that at the discontinuity, which approximately corresponds to $R_1=2.2$ b/s/Hz, all the rates in the half-open interval $R_2\in[1.7,2.3)$ b/s/Hz are achievable but do not belong to the Pareto boundary as defined at the beginning of Section \ref{sec:main}. Similarly, all the rate pairs $R_2=3.4$ b/s/Hz and $R_1\in(0,0.6)$ b/s/Hz are achievable but do not belong to the Pareto boundary. Therefore, these two line segments are plotted in dashed lines in Fig. \ref{fig:RvsAlpha}.

The scenario corresponding to Fig. \ref{fig:RvsAlpha2} presents an SIR greater than one. Because of that, the transition from the power-limited to the interference-limited regions presents no discontinuity as there is no jump in the transmit power at $2P_1=\gamma_{2\bar{R}}$ (see Fig. \ref{fig:CircCoeff2}). Notice that the discontinuity in the maximum transmit power at $R_1=2.9$ b/s/Hz, which is due to $a_{12}$ falling below the threshold, does not cause a discontinuity in the Pareto boundary. As previously explained, this is because at that point $a_{12}=\max[\iota(\alpha),\rho(\alpha)]$ holds, so proper and improper signaling achieve the same rate. Therefore, $R_2$ changes smoothly even though the transmit power and circularity coefficients jump. 
\begin{figure}[t!]
\centering
\includegraphics[width=1\columnwidth]{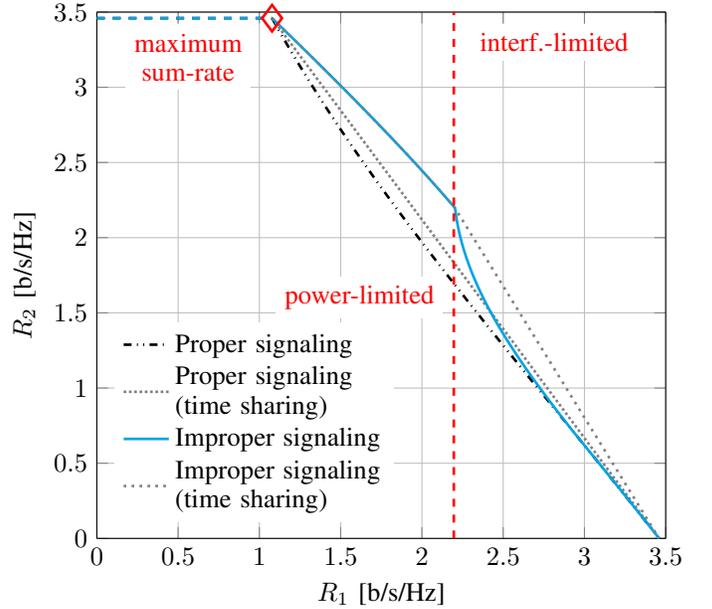}
\caption{Pareto boundary for proper and improper transmissions with $P_1=P_2=10$ and $a_{12}=0.8$.}
\label{fig:RvsAlpha2}
\end{figure}

Figures \ref{fig:RvsAlpha} and \ref{fig:RvsAlpha2} also illustrate that the dependency of $R_2$ on $R_1$ in the power-limited region is different from the interference-limited region. Specifically, $R_2$ decreases more slowly with $R_1$ in the former. This is due to the fact that in these examples $p_2=P_2$ holds in the entire power-limited region, and therefore the decrease in $R_2$ is due to an increase in $\kappa_2$. In the interference-limited region, however, $p_2$ decreases towards 0 as $R_1$ increases, which has a stronger effect on the achievable rate. The transition between these two different behaviors leads to a sharp bend in the Pareto boundary, as observed in Figs. \ref{fig:RvsAlpha} and \ref{fig:RvsAlpha2}. It corresponds to the transition between the power-limited and interference-limited regions only if $P_2>q(1)$ holds in the entire interference-limited region, and improper signaling is optimal in the power-limited region. This is because these conditions imply that $\kappa_2=1$ is optimal at the transition point $2P_1=\gamma_{2\bar{R}}$. Therefore, to satisfy the rate of the first user $p_2$ must decrease once we move into the interference-limited region, which changes the behavior of $R_2$. If $P_2<q(1)$ for some points in the interference-limited region, then $\kappa_2<1$ is optimal at $2P_1=\gamma_{2\bar{R}}$. Increasing $\kappa_2$ while keeping $p_2$ constant permits achieving the required $R_1$ when we move into the interference-limited region until $\kappa_2=1$ is reached. Hence the dependency of $R_2$ on $R_1$ will change at the value of $R_1$ such that $P_2=q(1)$, which is slightly shifted to the right. Figures \ref{fig:RvsAlpha} and \ref{fig:RvsAlpha2} also show the enlargement of the rate region due to improper signaling. In these examples, the interference level is significant, so the achievable rate region by improper signaling is substantially larger than that of proper signaling, especially in the scenario shown in Fig. \ref{fig:RvsAlpha}.

It is also worth highlighting that time sharing provides a substantial enlargement of the rate region for both proper and improper signaling, but there is still a significant gap between the performance of both strategies. In Fig. \ref{fig:RvsAlpha}, the rate region achieved by improper signaling is convexified by time sharing between the extreme points of the Pareto boundary, which correspond to proper signaling transmissions, and the point at the transition between the power-limited and the interference-limited regions, where both users transmit maximally improper signals. By contrast, all the Pareto boundary points in the power-limited region belong to the convex hull in Fig. \ref{fig:RvsAlpha2}, and time sharing improves the performance only in the interference-limited region.

\subsubsection{Operation regimes}
Considering the possibly suboptimal approach of always treating interference as noise, we may distinguish the following operation regimes that account for the optimality of improper signaling.
\begin{itemize}
	\item \textbf{Strictly improper regime:} If $a_{12}\geq\frac{P_1}{P_1+1}$, we say the Z-IC is in the strictly improper regime. In this regime and treating interference as noise, all the Pareto-optimal points satisfying $R_1>\log_2(1+\frac{P_1}{1+P_2a_{12}})$ are achieved by improper signaling, i.e., $\kappa_2>0$.
	\item \textbf{Selective improper/proper regime:} If $\frac{P_1}{P_1+1}>a_{12}>\min_{\alpha}\max[\iota(\alpha),\rho(\alpha)]$, we say the Z-IC is in the selective improper/proper regime. In this regime and treating interference as noise, only a subset of the Pareto-optimal points satisfying $R_1>\log_2(1+\frac{P_1}{1+P_2a_{12}})$ is achieved by improper signaling, i.e., $\kappa_2>0$.
	\item \textbf{Strictly proper regime:} If $a_{12}\leq\min_{\alpha}\max[\iota(\alpha),\rho(\alpha)]$, we say the Z-IC is in the strictly proper regime. In this regime and treating interference as noise, all the Pareto-optimal points are achieved by proper signaling, i.e., $\kappa_2=0$.
\end{itemize}
Note that $(R_1=\log_2(1+\frac{P_1}{1+P_2a_{12}}),R_2=\log_2(1+P_2))$ always belongs to the Pareto boundary and can only be achieved by proper signaling. Because of that, the operation regimes described above explain the optimal strategies for the boundary points satisfying $R_1>\log_2(1+\frac{P_1}{1+P_2a_{12}})$. The strictly improper regime is obtained making use of Lemma \ref{th:lemma2}, which establishes that $a_{12}\geq\mu(\alpha)=1-\frac{P_1}{\gamma_{2\bar{R}}-\gamma_{\bar{R}}}$ is a sufficient condition for improper signaling to be optimal if proper signaling does not allow maximum power transmission, i.e., if $R_1>\log_2(1+\frac{P_1}{1+P_2a_{12}})$. Therefore, $\mu(\alpha)\geq \max[\iota(\alpha),\rho(\alpha)]$ holds for all boundary points satisfying $R_1>\log_2(1+\frac{P_1}{1+P_2a_{12}})$. Furthermore, it can be easily checked that 
\begin{equation}
	\mu(1)=\iota(1)=\frac{P_1}{P_1+1} \, \Rightarrow \, \max\left[\iota(1),\rho(1)\right]=\frac{P_1}{P_1+1} \; .
\end{equation}
Since $\mu(\alpha)$ is an increasing function, the condition for the optimality of improper signaling is fulfilled in the strictly improper regime for all Pareto-optimal points except for $R_1=\log_2(1+\frac{P_1}{1+P_2a_{12}})$. An example of the operation in the strictly improper regime is given in Fig. \ref{fig:RvsAlpha} (and the corresponding Fig. \ref{fig:CircCoeff}). Figures \ref{fig:CircCoeff2} and \ref{fig:RvsAlpha2} correspond to the selective improper/proper regime, and we can observe that improper signaling is optimal only in the interval $R_1\in(1.1,2.9)$ b/s/Hz.

When we move from the power-limited region to the interference-limited region we have
\begin{align}
	&\lim_{\alpha\rightarrow\alpha_0^+}\iota(\alpha)=\frac{1}{4}\frac{\gamma_{\bar{R}}}{\gamma_{\bar{R}}+1}=\frac{1}{4}\left(1-\frac{1}{\sqrt{1+2P_1}}\right) \notag\\ &\Rightarrow \lim_{\alpha\rightarrow\alpha_0^+}\max\left[\iota(\alpha),\rho(\alpha)\right]\geq\frac{1}{4}\left(1-\frac{1}{\sqrt{1+2P_1}}\right) \; ,
\end{align}
where $\alpha_0$ is such that $2P_1=\gamma_{2\bar{R}}$. Therefore, proper signaling is the optimal strategy in the whole interference-limited region if $a_{12}\leq \frac{1}{4}(1-\frac{1}{\sqrt{1+2P_1}})$. In such a case, benefits of improper signaling, if any, are limited to the power-limited region. Notice that, as $P_1$ increases, this condition converges to $a_{12}\leq\frac{1}{4}$.

\subsection{Relationship to previous work}
Finally, we would like to connect our results to related works in the literature. In our previous work \cite{LameiroSantamariaSchreier:2015:Benefits-of-Improper-Signaling-for-Underlay}, we studied a similar scenario in the context of underlay cognitive radio. We considered the Z-IC with the restriction that the first user transmits only proper signals, i.e., $\kappa_1=0$. A characterization of the maximum achievable rate of user 2 was derived in terms of a threshold in $a_{12}$. Specifically, improper signaling was shown to be optimal if $a_{12}>1-\frac{P_1}{\gamma_{2\bar{R}}}$. This threshold is strictly higher than the one obtained for the general Z-IC (it can easily be seen that $1-\frac{P_1}{\gamma_{2\bar{R}}}\geq\mu(\alpha)\geq\iota(\alpha)$). This is in agreement with the fact that, if we let the first user optimize its circularity coefficient, the rate achieved by the second user can only increase.

The Z-IC was also considered in \cite{Kurniawan2015}, and the transmit strategy that maximizes the sum-rate was derived in closed form based on the real-composite model. Although such a model is usually more convenient from an optimization point-of-view, it is not as insightful as the augmented complex model since some of the features of the improper signal are not easily captured. This is the case for the degree of impropriety, which is elegantly given by the circularity coefficient. We would like to stress that in \cite{Kurniawan2015} only one point of the rate region was characterized, whereas in this work we completely characterized the boundary of the rate region. Furthermore, since we consider the augmented complex model, we provided closed-form formulas for the circularity coefficients, thus providing a more insightful description of how improper signaling behaves in this scenario. Nevertheless, some of the conclusions drawn in \cite{Kurniawan2015} fall within our characterization of the rate region boundary. By looking at the structure of the sum-rate maximizing transmit strategies in \cite[Eq. (31)]{Kurniawan2015}, we observe that improper signaling is chosen when $a_{12}>1$. This condition belongs to the strictly improper regime, and the solution presented in \cite[Eq. (31)]{Kurniawan2015} can be seen as a special case of \eqref{eq:kOpt}. Furthermore, according to \cite{Kurniawan2015}, when improper signaling is preferred for the sum-rate maximization, the signal transmitted by one of the users is chosen as maximally improper, and the selection of the user depends on whether or not the condition $P_2a_{12}\leq P_1$ holds. This is in agreement with our previous discussion, where we pointed out that the circularity coefficient of the first user is always equal to or smaller than that of the second user whenever the aforementioned condition is fulfilled. In such a case, the Pareto optimal point that corresponds to the sum-rate maximization is then the point that satisfies $q(1)=P_2$. Otherwise, if $P_2a_{12}> P_1$, the sum-rate is maximized when the first user transmits a maximally improper signal, and in this case the circularity coefficient of the second user is smaller than 1 except for $R_1=\frac{1}{2}\log_2(1+2P_1)$.

The MIMO Z-IC with improper signaling has been considered in \cite{Lagen2014,Lagen2016}. Due to the inclusion of the spatial dimension, an analytical characterization of improper signaling similar to that of Theorem \ref{th:theorem1} becomes intractable. Notice that in the SISO Z-IC, the impropriety is completely described by the circularity coefficient. However, in the MIMO case, the description of the impropriety is by means of the complementary covariance matrix, which is characterized by a set of circularity coefficients and a unitary matrix \cite{Schreier2010}. Because of that, the authors proposed a heuristic scheme to design the complementary covariance matrix of the transmitted signal of user 2 such that it permits easy control of the degree of impropriety. An interesting aspect of \cite{Lagen2014,Lagen2016} is the fact that it does not completely stick to a particular model for the representation of improper signals. Thus, they use the real-composite model to optimize the transmission scheme of user 1, while they design the scheme of the second user by means of the augmented-complex formulation. Interestingly, the achievable rate regions obtained in \cite{Lagen2014,Lagen2016} have a shape similar to those in Figs. \ref{fig:RvsAlpha} and \ref{fig:RvsAlpha2}. This suggests that some of our conclusions for the single-antenna Z-IC might be extended to the more general MIMO Z-IC.

\section{Conclusion}\label{sec:conc}
We have analyzed the benefits of improper signaling in the single-antenna Z-IC. Under the assumption that interference is treated as Gaussian noise, we have derived a complete and insightful characterization of the Pareto boundary of the rate region, and the corresponding transmit powers and circularity coefficients in closed-form. This characterization has been derived by analyzing how the circularity coefficients affect the performance at the different boundary points. Specifically, we have shown that improper signaling is optimal when the interference coefficient exceeds a given threshold that depends on the rate achieved by the interfered user. We have shown that the rate region can be substantially enlarged by using improper signaling, especially when the relative level of interference is high.

\section*{Acknowledgements}
The work of C. Lameiro and P. J. Schreier was supported by the German Research Foundation (DFG) under grant SCHR 1384/6-1. The work of I. Santamaria was supported by the Ministerio de Economia y Competitividad (MINECO), Spain, under projects RACHEL (TEC2013-47141-C4-3-R) and CARMEN (TEC2016-75067-C4-4-R).

The authors would like to thank the anonymous reviewers for their valuable comments, which have helped improve the quality of the paper.

\appendices

\section{proof of Lemma \ref{th:lemma0}}\label{app:lemma0}
Let us first consider $\kappa_1=1$. In this case $R_1(p_2,\kappa_2)\geq\bar{R}$ is equivalent, because of \eqref{eq:R1pk}, to the quadratic inequality
	\begin{align}
		p_2^2a_{12}^2\gamma_{2\bar{R}}\left(1-\kappa_2^2\right)&+p_22a_{12}\left[\gamma_{2\bar{R}}-P_1\left(1+\kappa_2\right)\right]\notag\\&+\gamma_{2\bar{R}}-2P_1\leq0 \; ,
	\end{align}
	which is convex and only has one positive root given by
	\begin{equation}
		p_2=\frac{P_1\left(1+\kappa_2\right)-\gamma_{2\bar{R}}+\left|P_1\left(1+\kappa_2\right)-\gamma_{2\bar{R}}\kappa_2\right|}{a_{12}\gamma_{2\bar{R}}\left(1-\kappa_2^2\right)} \; .\label{eq:p2k10}
	\end{equation}
	Notice that the rate constraint of user 1 can only be fulfilled with a maximally improper signal if $\bar{R}  \leq \frac{1}{2} \log_2(1+2P_1)$, that is, if user 1 can achieve its rate by transmitting all power over the real or imaginary part. This means that $2P_1\geq\gamma_{2\bar{R}}$ holds whenever $\kappa_1=1$, which implies $P_1(1+\kappa_2)-\gamma_{2\bar{R}}\kappa_2\geq0$. Hence, \eqref{eq:p2k10} can be simplified to $p_2=q(\kappa_2)$, where $q(\kappa_2)$ is given by the second branch of \eqref{eq:pk}. If $\kappa_1<1$, $R_1(p_2,\kappa_2)\geq\bar{R}$ is equivalent to the quadratic inequality
	\begin{align}
		p_2^2a_{12}^2&\left[\left(\gamma_{2\bar{R}}+1\right)\left(1-\kappa_2^2\right)-1\right]-p_22a_{12}\left(P_1-\gamma_{2\bar{R}}\right)\notag\\&+\gamma_{2\bar{R}}+1-\left(1+P_1\right)^2\leq0 \; .\label{eq:ineq}
	\end{align}
	The equivalent power constraint will then be given by one of the roots of this equation. In order to determine which one of the two roots must be considered we can use the fact that $R_1(p_2,\kappa_2)$ is decreasing in $p_2$. This means that at least one of the roots will violate the condition $0\leq\kappa_1<1$, with $\kappa_1$ given by \eqref{eq:kappa1}. The roots of \eqref{eq:ineq} are
	\begin{equation}\label{eq:roots}\resizebox{0.49\textwidth}{!}{$
	\begin{aligned}
		&p_2=\frac{1}{a_{12}\left[\left(\gamma_{2\bar{R}}+1\right)\left(1-\kappa_2^2\right)-1\right]}\Bigg[\left(P_1-\gamma_{2\bar{R}}\right)\pm\\ &\left.\sqrt{\left(P_1-\gamma_{2\bar{R}}\right)^2+\left[\left(\gamma_{2\bar{R}}+1\right)\left(1-\kappa_2^2\right)-1\right]\left[\left(1+P_1\right)^2-\gamma_{2\bar{R}}-1\right]}\right] \hspace{-0.1cm}.
	\end{aligned}$}	
	\end{equation}
	Notice that $(1+P_1)^2-\gamma_{2\bar{R}}-1=(1+P_1)^2-(1+P_1)^{2\alpha}\geq0$. Therefore, if \eqref{eq:ineq} is convex, i.e., if $[(\gamma_{2\bar{R}}+1)(1-\kappa_2^2)-1\geq0$, \eqref{eq:ineq} has one negative and one positive root, and the equivalent power constraint will be given by the latter if it satisfies $\kappa_1<1$. Otherwise, the rate expression for $\kappa_1=1$ must be considered. If \eqref{eq:ineq} is concave, both roots can be either positive, negative, or complex. For the last two cases, none of them satisfies $0\leq\kappa_1<1$, so the equivalent power constraint must be obtained through the expression for $\kappa_1=1$. If the two roots are positive, the monotonicity of $R_1(p_2,\kappa_2)$ in $p_2$ is not fulfilled for the largest root, hence the equivalent power constraint is determined by the smallest root. For all these cases, the root that must be considered is obtained by taking the positive square root in \eqref{eq:roots} which, after some manipulations, yields $p_2=q(\kappa_2)$, with $q(\kappa_2)$ being given by the first branch of \eqref{eq:pk}. This concludes the proof.

\section{Proof of Lemma \ref{th:lemma1}}\label{app:lemma1}
Through \eqref{eq:R2}, the derivative of $R_2(\kappa_2)$ with respect to $\kappa_2^2$ is non-negative if
\begin{equation}\label{eq:derR}
	\frac{\partial R_2(\kappa_2)}{\partial\kappa_2^2}\geq0 \; \Leftrightarrow \; 2\frac{\partial q(\kappa_2)}{\partial\kappa_2^2}\left[q(\kappa_2)(1-\kappa_2^2)+1\right]\geq q(\kappa_2)^2 \; .
\end{equation}
Let us first consider $\kappa_1=1$. In this case we have
\begin{equation}\label{eq:derq1}
	\frac{\partial q(\kappa_2)}{\partial\kappa_2^2}=\frac{q(\kappa_2)}{2\kappa_2(1-\kappa_2)} \; .
\end{equation}
Plugging \eqref{eq:derq1} into \eqref{eq:derR} we obtain
\begin{equation}
		\frac{\partial R_2(\kappa_2)}{\partial\kappa_2^2}\geq0 \; \Leftrightarrow \; q(\kappa_2)(1-\kappa_2)+1\geq0 \; ,
	\end{equation}
	which holds for all values of $\kappa_2$. If $\kappa_1<1$, \eqref{eq:ineq} holds with equality for $p_2=q(\kappa_2)$. Therefore, evaluating \eqref{eq:ineq} at $p_2=q(\kappa_2)$ and taking the derivative with respect to $\kappa_2^2$ yields
	\begin{equation}\label{eq:derqm1}\resizebox{0.48\textwidth}{!}{$\begin{aligned}
		\frac{\partial q(\kappa_2)}{\partial\kappa_2^2}=\frac{q(\kappa_2)^2a_{12}(\gamma_{2\bar{R}}+1)}{2q(\kappa_2)a_{12}\left[\left(\gamma_{2\bar{R}}+1\right)\left(1-\kappa_2^2\right)-1\right]-2\left(P_1-\gamma_{2\bar{R}}\right)} \; .\end{aligned}$}
	\end{equation}	
	Plugging \eqref{eq:derqm1} into \eqref{eq:derR} we obtain
	\begin{align}
		&\frac{\partial R_2(\kappa_2)}{\partial\kappa_2^2}\geq0 \; \Leftrightarrow \; \notag\\ &\frac{a_{12}(\gamma_{2\bar{R}}+1)\left[q(\kappa_2)(1-\kappa_2^2)+1\right]}{q(\kappa_2)a_{12}\left[\left(\gamma_{2\bar{R}}+1\right)\left(1-\kappa_2^2\right)-1\right]-\left(P_1-\gamma_{2\bar{R}}\right)}\geq 1 \; .
	\end{align}
	The denominator in this expression can be shown to be positive as follows. Since we are assuming that $\kappa_1<1$ holds, we have, through \eqref{eq:pk}, that $q(\kappa_2)a_{12}[(\gamma_{2\bar{R}}+1)(1-\kappa_2^2)-1]-(P_1-\gamma_{2\bar{R}})=\sqrt{(\gamma_{2\bar{R}}+1)[P_1^2(1-\kappa_2^2)+(\gamma_{2\bar{R}}-2P_1)\kappa_2^2]}$, hence it is positive. Consequently, we obtain
	\begin{equation}\label{eq:lemma1}
		\frac{\partial R_2(\kappa_2)}{\partial\kappa_2^2}\geq0 \; \Leftrightarrow \; a_{12}\geq\frac{\gamma_{2\bar{R}}-P_1-\bar{q}(\kappa_2)}{\gamma_{2\bar{R}}+1} \; ,
	\end{equation}
	where $\bar{q}(\kappa_2)=q(\kappa_2)a_{12}$, which does not depend on $a_{12}$ as can be seen from \eqref{eq:pk}. Let $\hat{\kappa}_2$ be such that the right-hand side of \eqref{eq:lemma1} holds with equality. Since $\frac{\partial q(\kappa_2)}{\partial\kappa_2^2}>0$ for $q(\kappa_2)>0$, $\bar{q}(\kappa_2)$ increases with $\kappa_2$ and, consequently, the right-hand side of \eqref{eq:lemma1} holds with strict inequality when $\kappa_2>\hat{\kappa}_2$. As a result, the derivative of $R_2(\kappa_2)$ is positive whenever $\kappa_2>\hat{\kappa}_2$, which concludes the proof.

\section{Proof of Lemma \ref{th:lemma2}}\label{app:lemma2}
By Lemma \ref{th:lemma1}, $R_2(\kappa_2)$ increases monotonically in $\kappa_2$ as long as $\frac{\partial R_2(\kappa_2)}{\partial \kappa_2^2}\geq 0$ for $\kappa_2=0$. Clearly, $\kappa_1$ is also zero at this point, thus we can apply \eqref{eq:lemma1}, yielding
\begin{align}
		 a_{12}\geq\frac{\gamma_{2\bar{R}}-P_1-\bar{q}(0)}{\gamma_{2\bar{R}}+1}&=\frac{\gamma_{2\bar{R}}-P_1-\left(\frac{P_1}{\gamma_{\bar{R}}}-1\right)}{\gamma_{2\bar{R}}+1}\notag\\&=1-\frac{P_1}{\gamma_{2\bar{R}}-\gamma_{\bar{R}}}=\mu_1(\alpha) \; ,
	\end{align}
	which proves the first case. In the second case, improper signaling outperforms proper signaling only for $\kappa_2>\tilde{\kappa}_2$. By Lemma \ref{th:lemma1} we also know that, if improper signaling outperforms proper signaling for $\kappa_2=\kappa_2'$, then the rate improvement will be strictly positive for $\kappa_2>\kappa_2'$. Therefore, when $R_2(1)>R_2(0)$ holds, there must be $\tilde{\kappa}_2$ such that $R_2(\tilde{\kappa}_2)=R_2(0)$ and $R_2(\kappa_2)>R_2(0)$ for all $\kappa_2>\tilde{\kappa}_2$. To evaluate when $R_2(1)>R_2(0)$ holds, we first consider the boundary points satisfying $2P_1\geq\gamma_{2\bar{R}}$. Since at these points user 1 can achieve its required rate with a maximally improper signal, i.e., with $\kappa_1=1$, we have that $q(1)=\infty$. Consequently, $R_2(1)>R_2(0)$ is always fulfilled for $a_{12}>0$, which yields the second branch of \eqref{eq:mu2}. For the boundary points satisfying $2P_1<\gamma_{2\bar{R}}$, $\kappa_1<1$ holds. Therefore, using the corresponding branch in \eqref{eq:pk} for $\kappa_2=1$ yields
\begin{equation}
	q(1)=\frac{1}{a_{12}}\left[\gamma_{2\bar{R}}-P_1-\sqrt{\left(\gamma_{2\bar{R}}+1\right)\left(\gamma_{2\bar{R}}-2P_1\right)}\right] \; .
\end{equation}
Using this expression and \eqref{eq:R2}, $R_2(1)>R_2(0)$ holds if
\begin{align}
	\frac{2}{a_{12}}&\left[\gamma_{2\bar{R}}-P_1-\sqrt{\left(\gamma_{2\bar{R}}+1\right)\left(\gamma_{2\bar{R}}-2P_1\right)}\right]\notag\\&>\frac{1}{a_{12}}\left(\frac{P_1}{\gamma_{\bar{R}}}-1\right)\left[\frac{1}{a_{12}}\left(\frac{P_1}{\gamma_{\bar{R}}}-1\right)+2\right] \; ,
\end{align}
which yields
\begin{equation}
	a_{12}>\frac{\left(\frac{P_1}{\gamma_{\bar{R}}}-1\right)^2}{2\left[\gamma_{2\bar{R}}-P_1-\sqrt{\left(\gamma_{2\bar{R}}+1\right)\left(\gamma_{2\bar{R}}-2P_1\right)}-\frac{P_1}{\gamma_{\bar{R}}}+1\right]} \; .
\end{equation}
Taking into account that $\gamma_{2\bar{R}}+1=\left(\gamma_{\bar{R}}+1\right)^2$, we have
\begin{equation}\resizebox{0.48\textwidth}{!}{$
\begin{aligned}
	a_{12}>&\frac{\left(P_1-\gamma_{\bar{R}}\right)^2}{\gamma_{\bar{R}}\left(\gamma_{\bar{R}}+1\right)\left[\left(\gamma_{2\bar{R}}-2P_1\right)+\gamma_{2\bar{R}}-2\gamma_{\bar{R}}-2\gamma_{\bar{R}}\sqrt{\gamma_{2\bar{R}}-2P_1}\right]} \\
	&=\frac{\left(P_1-\gamma_{\bar{R}}\right)^2}{\left(\gamma_{2\bar{R}}-\gamma_{\bar{R}}\right)\left(\sqrt{\gamma_{2\bar{R}}-2P_1}-\gamma_{\bar{R}}\right)^2}=\iota(\alpha) \; ,
\end{aligned}$}
\end{equation}
which yields the first branch of \eqref{eq:mu2}. Finally, if $a_{12}\leq\iota(\alpha)$ then we have $R_2(1)\leq R_2(0)$ and, by Lemma \ref{th:lemma1}, $R_2(\kappa_2)\leq R_2(0)$ for all $\kappa_2$. This yields the third case and concludes the proof.

\section{Proof of Theorem \ref{th:theorem1}}\label{app:theorem1}
By Lemma \ref{th:lemma2}, improper signaling can only be optimal if $a_{12}>\iota(\alpha)$. However, this condition is not sufficient, since it does not take the power budget constraint into account. Firstly, since increasing $\kappa_2$ can improve the rate only if this permits increasing the transmit power as well, we must have that $q(0)<P_2$. This is equivalent to
\begin{equation}\label{eq:thProp}
	q(0)=\frac{1}{a_{12}}\left(\frac{P_1}{\gamma_{\bar{R}}}-1\right)<P_2 \, \Leftrightarrow \, a_{12}>\frac{1}{P_2}\left(\frac{P_1}{\gamma_{\bar{R}}}-1\right) \; .
\end{equation}
If $a_{12}$ is greater than this quantity and also than $\iota(\alpha)$, we know by Lemma \ref{th:lemma2} that $\kappa_2=1$ is optimal if $q(1)\leq P_2$. Otherwise, it is clear that, if improper signaling is optimal, then the optimal circularity coefficient $\kappa_2$ will satisfy $p(\kappa_2)=P_2$. Let us denote this circularity coefficient as $\kappa_2=\kappa_{\max}$. Then, using \eqref{eq:R2} we have
\begin{align}
	R_2(\kappa_{\max}&)\geq R_2(0) \, \Leftrightarrow \, \notag\\&\kappa_{\max}^2\leq 1-\frac{1}{P_2^2}\left[1+\frac{1}{a_{12}}\left(\frac{P_1}{\gamma_{\bar{R}}}-1\right)\right]^2+\frac{2P_2+1}{P_2^2} \; .\label{eq:kcond}
\end{align} 
Therefore, we have to evaluate the condition for $\kappa_{\max}$, with $p(\kappa_{\max})=P_2$, fulfilling the above expression. Let us first consider the case $\kappa_1=1$. From \eqref{eq:pk}, we have that
\begin{equation}
	q(\kappa_{\max})=P_2 \, \Leftrightarrow \, \kappa_{\max}=1-\frac{1}{a_{12}P_2}\left(\frac{2P_1}{\gamma_{2\bar{R}}}-1\right) \; .
\end{equation}  
Plugging this expression into \eqref{eq:kcond} we have the condition
\begin{align}
	&\left[1-\frac{1}{a_{12}P_2}\left(\frac{2P_1}{\gamma_{2\bar{R}}}-1\right)\right]^2\leq \notag\\&1-\frac{1}{P_2^2}\left[1+\frac{1}{a_{12}}\left(\frac{P_1}{\gamma_{\bar{R}}}-1\right)\right]^2+\frac{2P_2+1}{P_2^2} \; ,
\end{align}
which, after some manipulations, yields the quadratic inequality
\begin{equation}
	a_{12}^2+a_{12}\left(P_2a_I-a_P\right)-\frac{1}{2}P_2\left(a_I^2+a_P^2\right)\geq 0 \; ,
\end{equation}
where $a_P$ and $a_I$ are respectively given by \eqref{eq:aP} and \eqref{eq:aI}. This quadratic expression is convex and has at most one positive root. Hence, the condition on $a_{12}$ is given by the largest root, which is
\begin{equation}\label{eq:thPow1}
	a_{12}>\frac{a_P-P_2a_I+\sqrt{(a_P-P_2a_I)^2+2P_2(a_I^2+a_P^2)}}{2} \; .
\end{equation} 
This expression, however, is only applicable if $\kappa_1=1$ holds for the above threshold. Using \eqref{eq:kappa1}, $\kappa_1=1$ holds if
\begin{equation}\label{eq:thPow2}
	\frac{a_P-P_2a_I+\sqrt{(a_P-P_2a_I)^2+2P_2(a_I^2+a_P^2)}}{2}\geq \frac{P_1}{P_2}+a_I \; .
\end{equation}
Otherwise we have $\kappa_1<1$, thus we have to obtain $\kappa_{\max}$ using the first branch in \eqref{eq:pk} (or, alternatively, taking equality in \eqref{eq:ineq} with $p_2=P_2$), which yields
\begin{equation}\resizebox{0.49\textwidth}{!}{$
\begin{aligned}
	&q(\kappa_{\max})=P_2 \, \Leftrightarrow \, \\&\kappa_{\max}^2=\frac{P_2a_{12}\left[P_2a_{12}\gamma_{2\bar{R}}-2\left(P_1-\gamma_{2\bar{R}}\right)\right]+\gamma_{2\bar{R}}+1-\left(1+P_1\right)^2}{P_2^2a_{12}^2\left(\gamma_{2\bar{R}}+1\right)} \; .
\end{aligned}$}
\end{equation}
Plugging this expression into \eqref{eq:kcond} we obtain
\begin{align}
	&\frac{P_2a_{12}\left[P_2a_{12}\gamma_{2\bar{R}}-2\left(P_1-\gamma_{2\bar{R}}\right)\right]+\gamma_{2\bar{R}}+1-\left(P_1+1\right)^2}{P_2^2a_{12}^2\left(\gamma_{2\bar{R}}+1\right)}\leq \notag\\&1-\frac{1}{P_2^2}\left[1+\frac{1}{a_{12}}\left(\frac{P_1}{\gamma_{\bar{R}}}-1\right)\right]^2+\frac{2P_2+1}{P_2^2} \; ,
\end{align}
which, after some manipulations, yields
\begin{align}\label{eq:Pineq}
	&P_2^2\frac{1}{\gamma_{2\bar{R}}+1}+2P_2\left[\frac{P_1-\gamma_{2\bar{R}}}{a_{12}\left(\gamma_{2\bar{R}}+1\right)}+1\right]\notag\\&+\frac{\left(1+P_1\right)^2-\gamma_{2\bar{R}}-1}{a_{12}^2\left(\gamma_{2\bar{R}}+1\right)}+1-\left[1+\frac{1}{a_{12}}\left(\frac{P_1}{\gamma_{\bar{R}}}-1\right)\right]^2\geq 0 \; .
\end{align}
Since the above quadratic expression is convex, we have to consider the largest root. As expected, one of the roots is
\begin{equation}
	P_2=\frac{1}{a_{12}}\left(\frac{P_1}{\gamma_{\bar{R}}}-1\right) \; ,
\end{equation}
which corresponds to the maximum power allowed by proper signaling. Since we have already considered this threshold in \eqref{eq:thProp}, we have to use the second root, which can be simplified to
\begin{equation}
	P_2=\frac{1}{a_{12}}\left\{2\left(\gamma_{2\bar{R}}+1\right)\left(1-a_{12}\right)-\left[\frac{P_1\left(2\gamma_{\bar{R}}+1\right)}{\gamma_{\bar{R}}}+1\right]\right\} \; .
\end{equation}
If the above root is the largest one, the inequality \eqref{eq:Pineq} is satisfied when $P_2$ exceeds this value, which yields the condition
\begin{equation}\label{eq:thPow3}
	a_{12}>\frac{\gamma_{\bar{R}}\left(2\gamma_{2\bar{R}}+1\right)-P_1\left(2\gamma_{\bar{R}}+1\right)}{\gamma_{\bar{R}}\left[P_2+2\left(\gamma_{2\bar{R}}+1\right)\right]} \; .
\end{equation}
Combining \eqref{eq:thProp}-\eqref{eq:thPow3} we obtain that $a_{12}>\rho(\alpha)$ must hold for improper signaling to be optimal, with $\rho(\alpha)$ given by \eqref{eq:rho}. However, this condition is only valid if $a_{12}>\iota(\alpha)$. Consequently, improper signaling will be optimal if and only if $a_{12}$ is greater than the dominating threshold, i.e., if and only if $a_{12}>\max[\iota(\alpha),\rho(\alpha)]$. If this expression is satisfied, then by Lemma \ref{th:lemma2} $\kappa_2$ must be increased until $p(\kappa_2)=P_2$, thus the optimal circularity coefficient is given by \eqref{eq:kOpt}. This concludes the proof.

\bibliography{ZchannelImproperPU}
\bibliographystyle{IEEEtran}

\end{document}